\documentclass[runningheads]{llncs}
\usepackage{xcolor}
\usepackage{booktabs}
\usepackage{algorithm}
\usepackage[noend]{algpseudocode}

\usepackage{lineno}
\usepackage{setspace}
\algnewcommand\algorithmicforeach{\textbf{for each}}
\algdef{S}[FOR]{ForEach}[1]{\algorithmicforeach\ #1\ \algorithmicdo}

\usepackage[utf8]{inputenc}
\usepackage{graphicx}
\begin{document}
\title{An Index for Sequencing Reads Based on The Colored de Bruijn Graph\thanks{Partially supported by a Conicyt Ph.D. Scholarship and by the European Union’s Horizon 2020 research and innovation programme under the Marie Sklodowska-Curie [grant agreement No 690941]}}
%
%
\author{Diego Díaz-Domínguez\inst{1,2}}
\authorrunning{D. Díaz-Domínguez}
%
\institute{Department of Computer Science, University of Chile, Chile
\and CeBiB — Center for Biotechnology and Bioengineering, University of Chile, Chile
\email{diediaz@dcc.uchile.cl}\\}
\maketitle 
\begin{abstract}
In this article, we show how to transform a colored de Bruijn graph (dBG) into a practical index for processing massive sets of sequencing reads. Similar to previous works, we encode an instance of a colored dBG of the set using $BOSS$ and a color matrix $C$. To reduce the space requirements, we devise an algorithm that produces a smaller and more sparse version of $C$. The novelties in this algorithm are (i) an incomplete coloring of the graph and (ii) a greedy coloring approach that tries to reuse the same colors for different strings when possible. We also propose two algorithms that work on top of the index;  one is for reconstructing reads, and the other is for contig assembly. Experimental results show that our data structure uses about half the space of the plain representation of the set (1 Byte per DNA symbol) and that more than 99\% of the reads can be reconstructed just from the index.
\end{abstract}
\section{Introduction} \label{sec:intro}

A set of \emph{sequencing reads} is a massive collection $R = \{R_1,\ldots,R_n \}$ of $n$ overlapping short strings that together encode the sequence of a DNA sample. Analyzing this kind of data allows scientists to uncover complex biological processes that otherwise could not be studied.  There are many ways for extracting information from a set of reads (see \cite{reuter2015high} for review).  However, in most of the cases, the process can be reduced to build a \emph{de Bruijn graph} (dBG) of the collection and then search for graph paths that spell segments of the source DNA (see  ~\cite{bray2016near,iqbal2012novo,salmela2016accurate} for some examples).

Briefly, a dBG is a directed labeled graph that stores the transitions of the substrings of size $k$, or \emph{kmers}, in $R$. Constructing it is relatively simple, and the resulting graph usually uses less space than the input text.  Nevertheless,  this data structure is lossy, so it is not always possible to know if the label of a path matches a substring of the source DNA. The only paths that fulfill this property are those in which all nodes, except the first and last, have indegree and outdegree one~\cite{kececioglu1995comb}. Still, they represent just a fraction of the complete dBG.
 
More branched parts of the graph are also informative, but traverse them requires extra information to avoid spelling incorrect sequences. A simple solution is to augment the dBG with colors, in other words, we assign a particular color $c_i$ to every string $R_i \in R$, and then we store the same $c_i$ in every edge that represents a kmer of $R_i$. In this way, we can walk over the graph always following the successor node colored with the same color of the current node.

The idea of coloring dBGs was first proposed by Iqbal et al. \cite{iqbal2012novo}. Their data structure, however, contemplated a union dBG built from several string collections, with colors assigned to the collections rather than particular strings. Considering the potential applications of colored dBGs, Boucher et al. \cite{BBGPS15} proposed a succinct version of the data structure of Iqbal et al. In their index, called VARI, the topology of the graph is encoded using $BOSS$~\cite{BOSS12}, and the colors are stored separately from the dBG in a binary matrix $C$, in which the rows represent the kmers and the columns represent the colors. Since the work of Boucher et al., some authors have tried to compress and manipulate $C$ even further; including that of \cite{almodaresi2017rainbowfish}, \cite{pandey2018mantis}, \cite{holley2015bloom}, while others, such as \cite{mustafa2017metannot} and \cite{mustafa2018dynamic} have proposed methods to store compressed and dynamic versions $C$.

An instance of a colored dBG for a single set $R$ can also be encoded using a color matrix. The only difference though is that the number of columns is proportional to the number of sequences in $R$. Assigning a particular color to every sequence is not a problem if the collection is of small or moderated size. However, massive datasets are rather usual in Bioinformatics, so even using a succinct representation of $C$ might not be enough. One way to reduce the number of columns is to reuse colors for those sequences that do not share any kmer in the dBG. Alpanahi et al. \cite{alipanahi2018recoloring} addressed this problem, and showed that it is unlikely that the minimum-size coloring can be approximated in polynomial time.

Alpanahi et al. also proposed a heuristic for recoloring the colored dBG of a set of sequences that, in practice, dramatically reduces the number of colors when $R$ is a set of sequencing reads. Their coloring idea, however, might still produce incorrect sequences, so the applications of their version of the colored dBG are still limited.

\noindent{\bf Our Contributions.} 
In this article, we show how to use a colored dBG to store and analyze a collection of sequencing reads succinctly. Similarly to VARI, we use $BOSS$ and the color matrix $C$ to encode the data. However, we reduce the space requirements by partially coloring the dBG and greedily reusing the same colors for different reads when possible. We also propose two algorithms that work on top of the data structure, one for reconstructing the reads directly from the dBG and other for assembling contigs. We believe that these two algorithms can serve as a base to perform Bioinformatics analyses in compressed space. Our experimental results show that on average, the percentage of nodes in $BOSS$  that need to be colored is about 12.4\%, the space usage of the whole index is about half the space of the plain representation of $R$ (1 Byte/DNA symbol), and that more than 99\% of the original reads can be reconstructed from the index.

\section{Preliminaries}\label{sec:pre}

\paragraph{\textbf{{\em DNA strings}}. \rm A DNA sequence $R$ is a string over the alphabet $\Sigma=\{\texttt{a},\texttt{c},\texttt{g},\texttt{t}\}$ (which we map to $[2..5]$), where every symbol represents a particular nucleotide in a DNA molecule. The DNA \emph{complement} is a permutation $\pi[2..\sigma]$ that reorders the symbols in $\Sigma$ exchanging \texttt{a} with \texttt{t} and \texttt{c} with \texttt{g}. The \emph{reverse complement} of $R$, denoted $R^{rc}$, is a string transformation that reverses $R$ and then replaces every symbol $R[i]$ by its complement $\pi(R[i])$. For technical convenience we add to $\Sigma$ the so-called \emph{dummy} symbol \texttt{\$}, which is always mapped to 1.}

\paragraph{\textbf{{\em de Bruijn graph}}. \rm A de Bruijn graph (dBG) \cite{de1946combinatorial} of the string collection $R=\{R_1,\ldots,R_n\}$ is a labeled directed graph $G=(V,E)$ that encodes the transitions between the substrings of size $k$ of $R$, where $k$ is a parameter. Every node $v \in V$ is labeled with a unique $k-1$ substring of $R$. Two nodes $v$ and $u$ are connected by a directed edge $(v,u) \in E$ if the $k-2$ suffix of $v$ overlaps the $k-2$ prefix of $u$ and the $k$-string resulted from the overlap exists as substring in $R$. The label of the edge is the last symbol of the label of node $u$.}

\paragraph{\textbf{{\em Rank and select data structures}}. \rm Given a sequence $B[1..n]$ of elements over the alphabet $\Sigma=[1..\sigma]$, $\texttt{rank}_{b}(B,i)$ with $i \in [1..n]$ and $b\in\Sigma$, returns the number of times the element $b$ occurs in $B[1..i]$, while $\texttt{select}_b(B,i)$ returns the position of the $i$th occurrence of $b$ in $B$. For binary alphabets, $B$ can be represented in $n+o(n)$ bits so that \texttt{rank} and \texttt{select} are solved in constant time~\cite{Cla96}. When $B$ has $m \ll n$ 1s, a compressed representation using $m\lg\frac{n}{m}+\mathcal{O}(m)+o(n)$ bits, still solving the operations in constant time, is of interest~\cite{raman2007}.}

\paragraph{\textbf{{\em BOSS index.}}}The $BOSS$ data structure \cite{BOSS12} is a succinct representation for dBGs based on the \emph{Burrows-Wheeler Transform} ($BWT$)~\cite{BW94}. In this index, the labels of the nodes are regarded as rows in a matrix and sorted in reverse lexicographical order, i.e., strings are read from right to left. Suffixes and prefixes in $R$ of size below $k-1$ are also included in the matrix by padding them with \texttt{\$} symbols in the right size (for suffixes) or the left side (for prefixes). These padded nodes are also called \emph{dummy}. The last column of the matrix is stored as an array $K[1..\sigma]$, with $K[i]$ being the number of labels lexicographically smaller than any other label ending with character $i$. Additionally, the symbols of the outgoing edges of every node are sorted and then stored together in a single array $E$. A bit vector $B[1..|E|]$ is also set to mark the position in $E$ of the first outgoing symbol of each node.
The complete index is thus composed of the vectors $E$, $K$, and $B$. It can be stored in $|E|(\mathcal{H}_{0}(E)+\mathcal{H}_{0}(B))(1+o(1)) + \mathcal{O}(\sigma \log n)$ bits, where $\mathcal{H}_0$ is the zero-order empirical entropy~\cite[Sec~2.3]{navarro2016compact}. This space is reached with a Huffman-shaped Wavelet Tree \cite{MN05} for $E$, a compressed bitmap~\cite{raman2007} for $B$ (as it is usually very dense), and a plain array for $K$. 
Bowe et al. \cite{BOSS12} defined the following operations over $BOSS$ to navigate the graph:
\begin{itemize}
    \item \texttt{outdegree$(v)$}: number of outgoing edges of $v$.
    \item \texttt{forward$(v, a)$}: node reached by following an edge from $v$ labeled with $a$.
    \item \texttt{indegree$(v)$}: number of incoming edges of $v$.
    \item \texttt{backward$(v)$}: list of the nodes with an outgoing edge to $v$.
    \item \texttt{nodeLabel$(v)$}: label of node $v$.
\end{itemize}
The first four operations can be answered in $\mathcal{O}(\log\sigma)$ time while the last one takes $\mathcal{O}(k\log\sigma)$ time. For our purposes, we also define the following operations:
\begin{itemize}
    \item \texttt{forward\_r$(v, r)$}: node reached by following the \emph{r-th} outgoing edge of $v$ in lexicographical order. 
    \item \texttt{label2Node$(S)$}: identifier in $BOSS$ of the node labeled with the $(k-1)$-string $S$. 
\end{itemize}
The function \texttt{forward\_r} is a small variation \texttt{forward}, and it maintains the original time, while the function \texttt{label2Node} is the opposite of \texttt{nodeLabel}, but it also maintain its complexity in $\mathcal{O}(k\log\sigma)$ time.

\paragraph{\textbf{{\em Graph coloring.}}}The problem of coloring a graph $G=(V,E)$ consists of assigning an integer $c(v) \in [1..\omega]$ to each node $v \in V$ such that i) no adjacent nodes have the same color and ii) $\omega$ is minimal. The coloring is \emph{complete} if all the nodes of the graph are assigned with one color, and it is \emph{proper} if constraint i) is met for each node. The chromatic number of a graph $G$, denoted by $\chi(G)$, is the minimum number of colors required to generate a coloring that is complete and proper. A coloring using exactly $\chi(G)$ colors is considered to be optimal. Determining if there is a feasible $\omega$-coloring for $G$ is well known to be NP-complete, while the problem of inferring $\chi(G)$ is NP-hard~\cite{lewis2015guide}.  

\paragraph{\textbf{{\em Colored dBG.}}}\label{prem:dbgcol} The first version of the colored dBG \cite{iqbal2012novo} was described as a union graph $G$ built from several dBGs of different string collections. The edges in $G$ that encode the kmers of the \emph{i-th} collection are assigned the color $i$. The compacted version of this graph \cite{BBGPS15} represents the topology of $G$ using the $BOSS$ index and the colors using a binary matrix $C$, where the position $C[i,j]$ is set to true if the kmer represented by the \emph{i-th} edge in the ordering of $BOSS$ is assigned color $j$. The rows of $C$ are then stored using the compressed representation for bit vectors of \cite{raman2007}, or using Elias-Fano encoding \cite{fano1971number,elias1974,okanohara2007practical} if the rows are very sparse. In the single set version of the colored dBG, the colors are assigned to every string. Therefore, the number of columns in $C$ grows with the size of the collection. Alipanahi et al.~\cite{alipanahi2018recoloring} noticed that we could reduce the space of $C$ by using the same colors in those strings that have no common kmers. This new problem was named the \emph{CDBG-recoloring}, and formally stated as follows; given a set $R$ of strings and its dBG $G$, find the minimum number of colors such that i) every string $R_i \in R$ is assigned one color and ii) strings having two or more kmers in common in $G$ cannot have the same color.  Alipanahi et al. \cite{alipanahi2018recoloring} showed that an instance $I(G')$ of the \emph{Graph-Coloring} problem can be reduced in polynomial time to another instance $I'(G)$ of the \emph{CDBG-recoloring} problem. Thus, any algorithm that finds $\chi(G')$, also finds the minimum number of colors for dBG $G$. However, they also proved that the decision version of \emph{CDBG-Recoloring} is NP-complete.

\section{Definitions}

Let $R=\{R_1,R_2,....,R_n\}$ be a collection of $n$ DNA sequencing \emph{reads}, and let $R'=\{R_1,R_1^{rc}..R_n,R_n^{rc}\}$ be a collection of size $2n$ that contains the strings in $R$ along with their DNA reverse complements. The dBG of order $k$ constructed from $R'$ is defined as $G_{R'}^{k}=(V,E)$, and an instance of $BOSS$ for $G_{R'}^{k}$ is denoted as $BOSS(G_{R'}^{k})=(V',E')$, where $V'$ and $E'$ include the dummy nodes and their edges. For simplicity, we will refer to $BOSS(G_{R'}^{k})$ just as $BOSS(G)$. A node in $V'$ is considered an \emph{starting} node if its $k-1$ label is of the form \texttt{\$}$A$, where \texttt{\$} is a dummy symbol and $A$ is a $k-2$ prefix of one or more sequences in $R'$. Equivalently, a node is considered an \emph{ending} node if its $k-1$ label is of the form $A$\texttt{\$}, with \texttt{\$} being a dummy and $A$ being a $k-2$ suffix of one or more sequences in $R'$. Nodes whose labels do not contain dummy symbols are considered \emph{solid}, and solid nodes with at least one predecessor node with outdegree more than one are considered \emph{critical}. For practical reasons, we define two extra functions, \texttt{isStarting} and \texttt{isEnding} that are used to check if a node is starting or ending respectively.

A \emph{walk} $P$ over the dBG of $BOSS(G)$ is a sequence $(v_0,e_0,v_1...v_{t-1},e_t,v_t)$ where $v_0,v_1,...v_{t-1},v_t$ are nodes and $e_1..e_t$ are edges, $e_i$ connecting $v_{i-1}$ with $v_i$. $P$ is a \emph{path} if all the nodes are different, except possibly the first and the last. In such case, $P$ is said to be a \emph{cycle}. A sequence $R_i \in R$ is \emph{unambiguous} if there is a path in $BOSS(G)$ whose label matches the sequence of $R_i$ and if no pair of colored nodes in $(u,v) \in P$ share a predecessor node $v' \in P$. In any other case, $R_i$ is \emph{ambiguous}. Finally, the path $P_i$ that spells the sequence of $R_i$ is said to be \emph{safe} if every one of its branching nodes has only one successor colored with the color of $R_i$.

We assume that $R$ is a \emph{factor-free} set, i.e., no $R_i \in R$ is also a substring of another sequence $R_j$, with $i\neq j$. 

\section{Coloring a dBG of reads}

In this section, we define a coloring scheme for $BOSS(G)$ that generates a more succinct color matrix, and that allows us to reconstruct and assemble unambiguous sequences of $R'$. We use the dBG of $R'$ because most of the Bioinformatic analyses require the inspection of the reverse complements of the reads.   
Unlike previous works, the rows in $C$ represent the nodes in $BOSS(G)$ instead of the edges.

\paragraph{\textbf{{\em A partial coloring.}}}\label{sec:col_sch}
We make $C$ more sparse by coloring only those nodes in the graph that are \emph{strictly} necessary for reconstructing the sequences. We formalize this idea with the following lemma:

\begin{lemma}\label{lem:col}
For the path of an unambiguous sequence $R_i \in R'$ to be safe we have to color the starting node $s_i$ that encodes the $k-2$ prefix of $R_i$, the ending node $e_i$ that encodes the $k-2$ suffix of $R_i$ and the critical nodes in the path.
\end{lemma}

\begin{proof}
We start a walk from $s_i$ using the following rules: (i) if the current node $v$ in the walk has outdegree one, then we follow its only outgoing edge, (ii) if $v$ is a branching node, i.e., it has outdegree more than one, then we inspect its successor nodes and follow the one colored with the same color of $s_i$ and (iii) if $v$ is equal to $e_i$, then we stop the traversal.\qed 
\end{proof}

Note that the successor nodes of a branching node are critical by definition, so they are always colored. On the other hand, nodes with outdegree one do not require a color inspection because they have only one possible way out. 

Coloring the nodes $s_i$ and $e_i$ for every $R_i$ is necessary; otherwise, it would be difficult to know when a path starts or ends. Consider, for example, using the solid nodes that represent the $k-1$ prefix and the $k-1$ suffix of $R_i$ as starting and ending points respectively. It might happen that the starting point of $R_i$ can also be a critical point of another sequence $R_j$. If we start a reconstruction from $s_i$ and pick the color of $R_j$, then we will generate an incomplete sequence. A similar argument can be used for ending nodes. The concepts associated with our coloring idea are depicted in Figures~\ref{fig:col_diag}A and \ref{fig:col_diag}B.

\paragraph{\textbf{{\em Unsafe coloring.}}}\label{sec:unreco} As explained in Section~\ref{prem:dbgcol}, we can use the recoloring idea of \cite{alipanahi2018recoloring} to reduce the number of columns in $C$. Still, using the same colors for unrelated strings is not safe for reconstructing unambiguous sequences. 

\begin{lemma}\label{lem:unco}
Using the same color $c$ for two unambiguous sequences $R_i,R_j \in R'$ that do not share any $k-1$ substring might produce an unsafe path for $R_i$ or $R_j$.
\end{lemma}

\begin{proof}
Assume there is another pair of sequences $R_x, R_y \in R'$ that do not share any $k-1$ subsequence either, to which we assign them color $c'$. Suppose that the paths of $R_x$ and $R_j$ crosses the paths of $R_i$ and $R_j$ such that the resulting dBG topology resembles a grid. In other words, if $R_i$ has the edge $(u,u')$ and $R_j$ has the edge $(v,v')$, then $R_x$ has the edge $(u,v)$ and $R_y$ has the edge $(u',v')$. In this scenario, $v$ will have two successors, node $v'$ from the path of $R_j$ and some other node $v''$ from the path of $R_x$. Both $v'$ and $v''$ are critical by definition so they will be colored with $c$ and $c'$ respectively. The problem is that node $v'$ is also a critical node for $R_y$, so it will also have color $c'$. The reason is that $u'$, a node that precedes $v'$, appears in $R_i$ and $R_y$. As a consequence, the path of $R_x$ is no longer safe because one of its nodes ($v$ in this example) has to successors colored with $c'$. A similar argument can be made for $R_i$ and color $c$. Figure \ref{fig:unsafe_col} depicts the idea of this proof.
\qed \end{proof}

When spurious edges connect paths of unrelated sequences that are assigned the same color (as the in the proof of Lemma~\ref{lem:unco}), we can generate chimeric strings if, by error, we follow one of those edges. In the algorithm, we solve this problem by assigning different colors to those strings with sporadic edges, even if they do not share any $k-1$ substring. 

\begin{figure}[!t]
\centering
\includegraphics[width=0.8\linewidth]{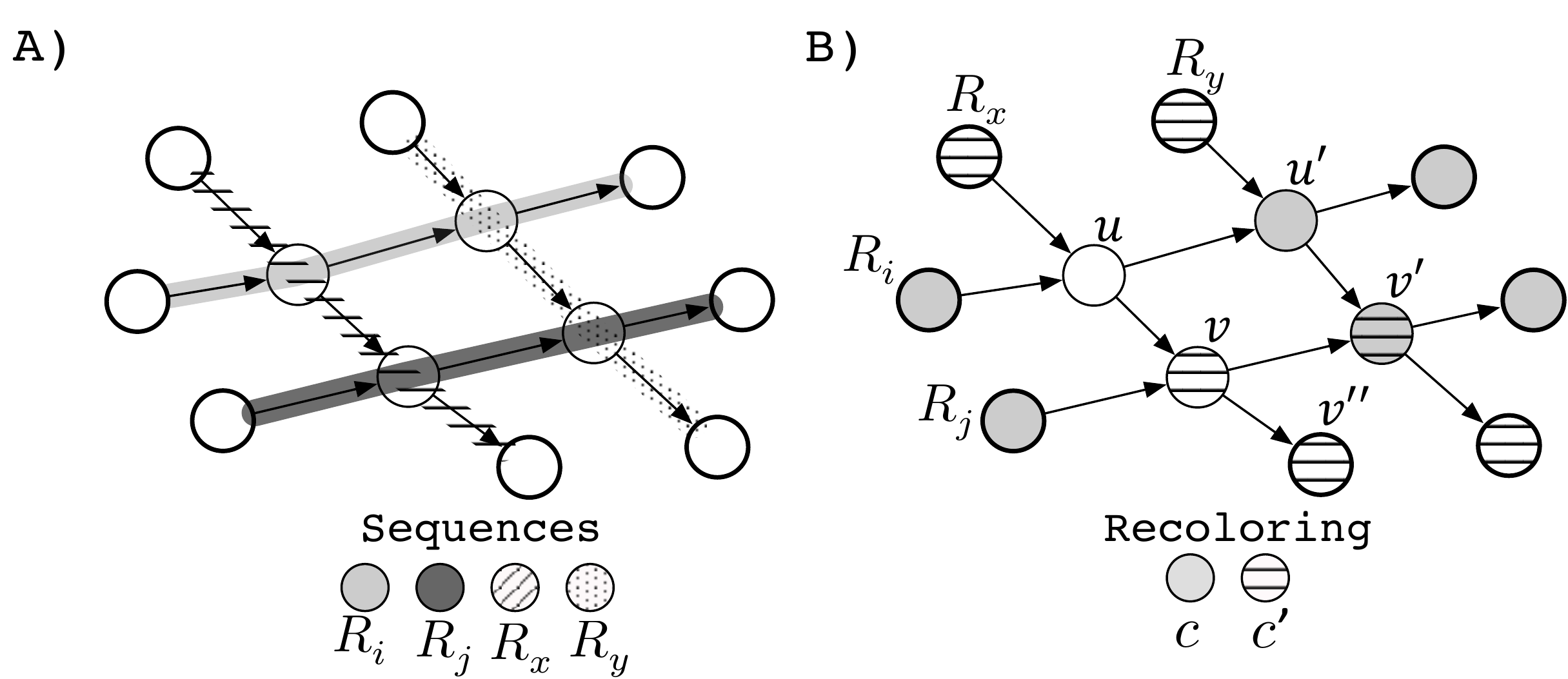}
\caption{Example of unsafe paths produced by a graph recoloring. (A) The dBG generated from the unambiguous sequences $R_i$, $R_j$, $R_x$ and $R_y$. Every texture represents the path of a specific string. (B) Recolored dBG. Sequences $R_i$ and $R_j$ are assigned the same color $c$ (light gray) as they do not share any $k-1$ substring. Similarly, sequences $R_x$ and $R_y$ are assigned another color $c'$ (horizontal lines) as they do not share any $k-1$ sequence neither. Nodes $u,u',v,v'$ and $v''$ are those mentioned in the Proof of Lemma ~\ref{lem:unco}. The sequences of $R_i$ and $R_x$ cannot be reconstructed as their paths become unsafe after the recoloring.}
\label{fig:unsafe_col}
\end{figure}

\paragraph{\textbf{{\em Safer and greedy coloring.}}}\label{sec:greedycol} Our greedy coloring algorithm starts by marking in a bitmap $N=[1..|V'|]$ the $p$ nodes of $BOSS(G)$ that need to be colored (starting, ending and critical). After that, we create an array $M$ of $p$ entries. Every $M[j]$ with $j \in [1..p]$ will contain a dynamic vector that stores the colors of the \emph{j-th} colored node in the $BOSS$ ordering. We also add $\texttt{rank}_1$ support to $N$ to map a node $v \in V$ to its array of colors in $M$. Thus, its position can be inferred as $\texttt{rank}_1(N,v)$.

The only inputs we need for the algorithm are $N$, $R'$ and $BOSS(G)$. For every $R_i \in R'$ we proceed as follows; we append a dummy symbol at the ends of the string, and then use the function \texttt{label2Node} to find the node $v$ labeled with the $k-1$ prefix of $R_i$. Note that this prefix will map a starting node as we append dummies to $R_i$. From $v$, we begin a walk on the graph and follow the edges whose symbols match the characters in the suffix $R_i[k..|R_i|]$. Note now that the last node $v'$ we visit in this walk is an ending node that maps the $k-1$ suffix of $R_i$. As we move through the edges, we store in an array $W_i$ the starting, ending, and critical nodes associated with $R_i$. Additionally, we push into another array $I_i$ the neighbor nodes of the walk that need to be inspected to assign a color to $R_i$. The rules for pushing elements into $I_i$ are as follows; i) if $v$ is a node in the path of $R_i$ with outdegree more than one, then we push all its successor nodes into $I_i$, ii) if $v$ is a node in the path of $R_i$ with indegree more than one, then we visit every predecessor node $v'$ of $v$, and if $v'$ has outdegree more than one, then we push into $I_i$ the successor nodes of $v'$. Once we finish the traversal, we create a hash map $H_i$ and fill it with the colors that were previously assigned to the nodes in $I_i$ and $W_i$. After that, we pick the smallest color $c'$ that is not in the keys of $H_i$, and push it to every array $M[\texttt{rank}_1(N,j)]$ with $j \in W_i$. After we process all the sequences in $R'$, the final set of colors is represented by the values in $M$. The whole processing of coloring a $R_i$ is described in detail by the procedure \texttt{greedyCol} in Pseudocode~\ref{alg:gc}.   

The construction of the sets $W_i$ and $I_i$ is independent for every string in $R'$, so it can be done in parallel. However, the construction of the hash map $H_i$ and the assignment of the color $c'$ to the elements of $W_i$ has to be performed sequentially as all the sequences in $R'$ need concurrent access to $M$. 
\begin{figure}[!t]
\centering
\includegraphics[width=0.9\linewidth]{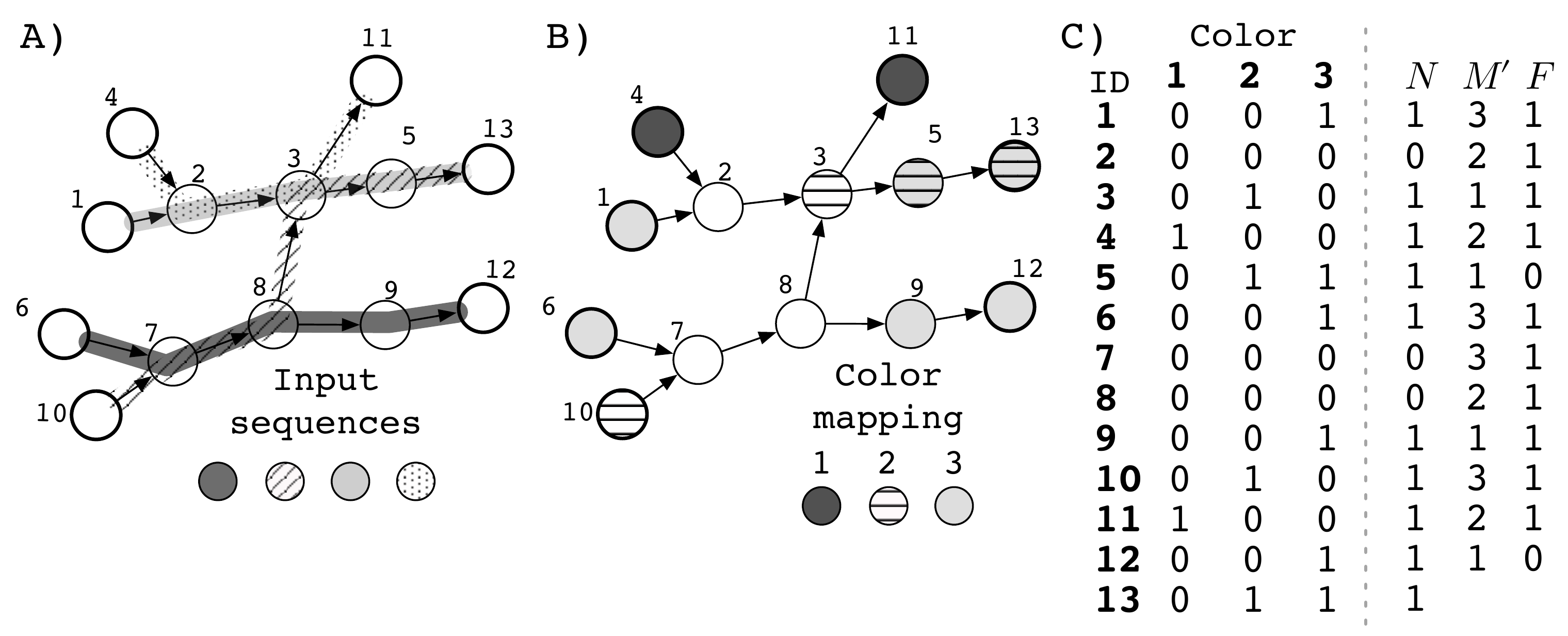}
\caption{Succinct colored dBG. (A) The topology of the graph. Colors and textures represent the paths that spell the input sequences of the dBG. Numbers over the nodes are their identifiers. Nodes 4,1,6 and 10 are starting nodes (darker borders). Nodes 11,13 and 12 are ending nodes and nodes 3,9,11 and 5 are critical. (B) Our greedy coloring algorithm. (C) The binary matrix $C$ that encodes the colors of Figure B. The left side is $C$ in its uncompressed form and the right side is our succinct version of $C$ using the arrays $N$,$M'$, and $F$.}
\label{fig:col_diag}
\end{figure}

\paragraph{\textbf{{\em Ambiguous sequences.}}}\label{sec:ambseqs} Our scheme, however, cannot safely retrieve sequences that are ambiguous. 

\begin{lemma}
Ambiguous sequences of $R'$ cannot be reconstructed safely from the color matrix $C$ and $BOSS(G)$.
\end{lemma}

\begin{proof}
Assume that collection $R$ is composed just by one string $R_1=XbXc$, where $X$ is a repeated substring and $b,c$ are two different symbols in $\Sigma$. Consider also that the kmer size for $BOSS(G)$ is $k=|X|+1$. The instance of $BOSS(G)$ will have a node $v$ labeled with $X$, with two outgoing edges, whose symbols are $b$ and $c$. Given our coloring scheme, the successor nodes of $v$ will be both colored with the same color. As a consequence, if during a walk we reach node $v$, then we will get stuck because there is not enough information to decide which is the correct edge to follow (both successor nodes have the same color).\qed
\end{proof}

A sequence $R_i$ will be ambiguous if it has the same $k-1$ pattern in two different contexts. Another case in which $R_i$ is ambiguous is when a spurious edge connects an uncolored node of $R_i$ with two or more critical nodes in the same path. Note that unlike unambiguous sequences with spurious edges, an ambiguous sequence will always be encoded by an unsafe path, regardless of the recoloring algorithm. In general, the number of ambiguous sequences will depend on the value we use for $k$.

\section{Compressing the colored dBG}\label{sec:colrep}

The pair $(M,N)$ can be regarded as a compact representation of $C$, where the empty rows were discarded. Every $M[i]$, with $i\in [1..|M|]$, is a row with at least one value, and every color $M[i][j]$, with $j \in [1..|M[i]|]$, is a column. However, $M$ is not succinct enough to make it practical. We are still using a computer word for every color of $M$. Besides, we need $|M|$ extra words to store the pointers for the lists in $M$. 

We compress $M$ by using an idea similar to the one implemented in $BOSS$ to store the edges of the dBG. The first step is to sort the colors of every list $M[i]$. Because the greedy coloring generates a set of unique colors for every node, each $M[i]$ becomes an array of strictly increasing elements after the sorting. Thus, instead of storing the values explicitly, we encode them as deltas, i.e., $M[i][j]=M[i][j]-M[i][j-1]$. After transforming $M$, we concatenate all its values into one single list $M'$ and create a bit map $F=[1..|M'|]$ to mark the first element of every $M[i]$ in $M'$. We store $M'$ using Elias-Fano encoding \cite{elias1974,fano1971number} and $F$ using the compressed representation for bit maps of \cite{raman2007}. Finally, we add $\texttt{select}_1$ support to $F$ to map a range of elements in $M'$ to an array in $M$. The complete representation of the color matrix now becomes $C = N + F + M'$ (see Figure \ref{fig:col_diag}C). The complete index of the colored dBG is thus composed of our version of $C$ and $BOSS(G)$. We now formalize the idea of retrieving the colors of a node from the succinct representation of $C$.

\begin{itemize}
    \item \texttt{getColors}$(v)$: list of colors assigned to node $v$.
\end{itemize}

\begin{theorem} the function  {\rm \tt getColors}$(v)$ computes in $\mathcal{O}(c)$ time the $c$ colors assigned to node $v$.
\end{theorem}

\begin{proof}
We first compute the rank $r$ of node $v$ within the colored nodes. This operation is carried out with $r=\texttt{rank}_1(N,v)$. After retrieving $r$, we obtain the range $[i..j]$ in $M'$ where the values of $v$ lie. For this purpose, we perform two $\texttt{select}_1$ operations over $F$, $[i,j]=(\texttt{select}_1(F,r),\texttt{select}_1(F,r+1)-1)$. Finally, we scan the range $[i..j]$ in $M'$, and as we read the values, we incrementally reconstruct the colors from the deltas. All the \texttt{rank} and \texttt{select} operations takes $\mathcal{O}(1)$, and reading the $c=j-i+1$ entries from $M'$ takes $\mathcal{O}(c)$, because retrieving an element from an Elias-Fano-encoded array also takes $\mathcal{O}(1)$. In conclusion, computing the colors of $v$ takes $\mathcal{O}(c)$.\qed   
\end{proof}

\section{Algorithms for the colored dBG}
\paragraph{\textbf{{\em Reconstructing unambiguous sequences.}}}\label{sec:trav_graph} We describe now an online algorithm that works on top of our index and that reconstructs all the unambiguous sequences in $R'$. We cannot tell, however, if a reconstructed string $R_i$ was present in the original set $R$ or if it was its reverse complement $R_{i}^{rc}$. This is not really a problem, because a sequence and its reverse complement are equivalent in most of the Bioinformatic analyses. 

The algorithm receives as input a starting node $v$. It first computes an array $A$ with the colors assigned to $v$ using the function \texttt{getColors} (see Section~\ref{sec:colrep}), and then sets a string $S=\texttt{nodeLabel}(v)$. For every color $a \in A$, the algorithm performs the following steps; initializes two temporary variables, an integer $v'=v$ and string $S'=S$, and then begins a graph walk from $v'$. If the outdegree of $v'$ is one, then the next node in the walk is the successor node $v'=\texttt{outgoing\_r}(v',1)$. On the other hand, if the outdegree of $v'$ is more than one, then the algorithm inspects all the successor nodes of $v'$ to check which one of them is the node $v''$ colored with $a$. If there is only one such $v''$, then it sets $v'=v''$. This procedure continues until $v'$ becomes an ending node. During the walk, the edge symbols are pushed into $S'$. When an ending node is reached, the algorithm reports $S'[1..|S'|-1]$ as the reconstructed sequence.

If at some point during a walk, the algorithm reaches a node with outdegree more than one, and with more than one successor colored with $a$, then aborts the reconstruction of the string as the path is unsafe for color $a$. Then, it returns to $v$ and continues with the next sequence. The complete procedure is detailed in the function \texttt{buildSeqs} of Algorithm~\ref{alg:bs}.

\paragraph{\textbf{{\em Assembling contigs.}}} Our coloring scheme for the dBG allows us to report sequences that represent the overlap of two or more strings of $R'$. There are several ways in which a set of sequences can be arranged such that they form valid overlaps, but in practice, we are not interested in all such combinations. What we want is to compute only those union strings that describe real segments of the underlying genome of $R'$, a.k.a \emph{contig} sequences. In this work we do not go deep into the complexities of contig assembly  (see \cite{kececioglu1995comb,medvedev2007computability,idury1995new,medvedev2011paired} for some review). Instead, we propose a simple heuristic, that work on top of our index, and that it is aimed to produce contigs that are longer than those produced by uncolored dBGs. 

Similar to \texttt{buildSeqs}, this method traverses the graph to reconstruct the contigs. During the process, it uses the color information to weight the outgoing edges \emph{on the fly}, and thus, inferring which is the most probable path that matches a real segment of the source DNA.  

The algorithm receives as input a starting node $v$ and initializes a set $L$ and hash map $Q$. Both data structures are used to store information about the strings that belong to the contig of $v$. A read $R_i \in R'$ is identified in the index as a pair $(c,v')$, where $c$ is a color assigned to $R_i$ and $v'$ is the starting node of its path. $L$ contains the reads already traversed while $Q$ contains the active reads. The algorithm also initializes a string $S=\texttt{nodeLabel}(v)$ and pushes every pair $Q[c_i]=v$ with $c_i \in \texttt{getColor}(v)$. After that, it begins a walk from $v$ and pushes into $S$ the symbols of the edges it visits. For every new node $v'$ reached during the walk, the algorithm checks if one of its predecessor nodes, say $u$, is a starting node. If so, then for every $c_i \in \texttt{getColors}(u)$ sets $Q[c_i]=u$ if $(c_i,u)$ does not exist in $L$. On the other hand, if one of the successors of $v'$, say $u'$, is an ending node, then for every $c_i \in \texttt{getColors}(u')$ sets $L[(c,Q[c_i])]$ and then removes the entry $Q[c_i]$. After updating $Q$ and $L$, it selects one of the outgoing edges of $v'$ to continue the walk. For this purpose, the algorithm uses the following rules; (i) if $v'$ has outdegree one, then it takes its only outgoing edge, (ii) if $v'$ has outdegree more than one, then it inspects how the colors in $Q$ distribute among the successors of $v'$. If there is only one successor node of $v'$, say $v''$, colored with at least $x$ fraction of the colors of $Q$, where $x$ is a parameter, then the algorithm follows $v''$, and removes from $Q$ the colors of the other successor nodes of $v'$.

The algorithm will stop if; (i) there is no such $v''$ that meet the $x$ threshold, (ii) there is more than one successor of $v'$ with the same color or (iii) $v'$ has outdegree one, but the successor node is an ending node. After finishing the walk, the substring $S[2..|S|]$ is reported as the contig. The procedure \texttt{contigAssm} in Pseudocode~\ref{alg:ca} describes in detail the contig assembly algorithm, and a graphical example is shown in Figure~\ref{fig:ctg_assm}.


\begin{figure}[!t]
\centering
\includegraphics[width=\linewidth]{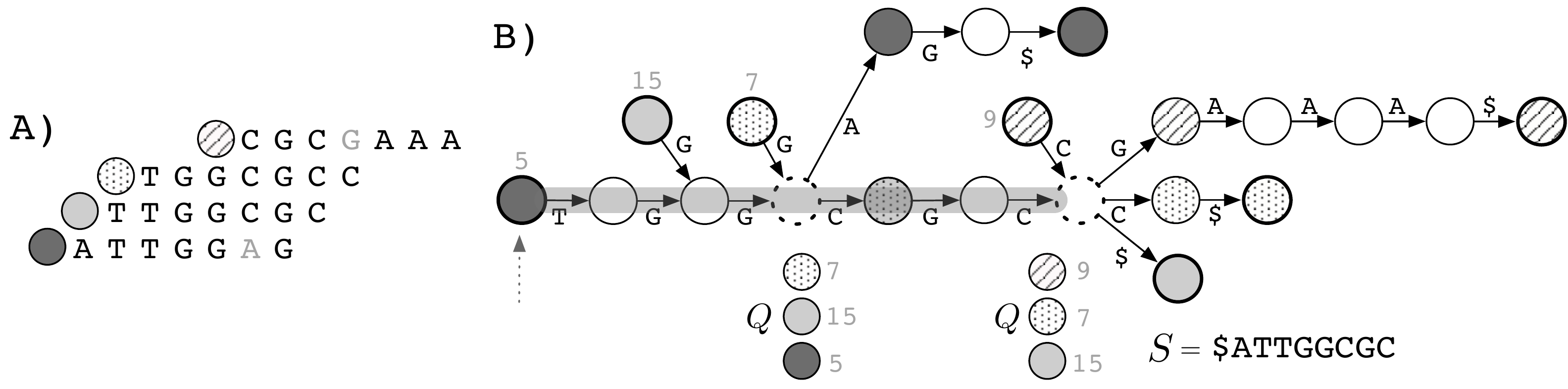}
\caption{Example of the assembly of a contig using our index. (A) Inexact overlap of four sequences. The circle to the left of every string represents its color in the dBG. Light gray symbols are mismatches in the overlap. (B) The colored dBG of the sequences. Circles with darker borders are starting and ending nodes. Light gray values over the starting nodes are their identifiers. The contig assembly begins in node 5 (denoted with a dashed arrow) and the threshold $x$ to continue the extension is set to 0.5. The state of the hash map $Q$ when the walk reaches a branching node (dashed circles) is depicted below the graph. The assembly ends in the right-most branching node as it has not a successor node that contains at least 50\% of the colors in $Q$. The final contig is shown as a light grey path over the graph, and its sequence is stored in $S$.}
\label{fig:ctg_assm}
\end{figure}

\section{Experiments}

We use a set of reads generated from the E.coli genome\footnote{http://spades.bioinf.spbau.ru/spades\_test\_datasets/ecoli\_mc} to test the ideas described in this article. The raw file was in \texttt{FASTQ} format and contained 14,214,324 reads of 100 characters long each. We preprocessed the file by removing sequencing errors using the tool of~\cite{bankevich2012spades}, and discarding reads with \texttt{N} symbols. The preprocessing yielded a data set of 8,655,214 reads (a \texttt{FASTQ} file of 2GB). Additionally, we discarded sequencing qualities and the identifiers of the reads as they are not considered in our data structure. From the resulting set $R$ (a text file of 833.67 MB), we created another set $R'$ that considers the elements in $R$ and their reverse complements.

Our version of the colored dBG, the algorithm for greedy coloring and the algorithms for reconstructing and assembly reads were implemented in \texttt{C++}\footnote{https://bitbucket.org/DiegoDiazDominguez/colored\_bos/src/master}, on top of the \texttt{SDSL-lite} library~\cite{gbmp2014sea}. In our implementation, arrays $M'$ and $F$ are precomputed beforehand to store the colors directly to them, because using the dynamic list $M$ is not cache-friendly. Additionally, all our code, except the algorithm for contig assembly, can be executed using multiple threads.

We built six instances of our index using $R'$ as input. We choose different values for $k$, from 25 to 50 in steps of five. The coloring of every one of these instances was carried out using eight threads. Statistics about the graph topologies are shown in Table \ref{tab:dbg}, and statistics about the coloring process are shown in Table 2. In every instance, we reconstructed the unambiguous reads (see  Table \ref{tab:color}). Additionally, we generated an FM-index of $R'$ to locate the reconstructed reads and check that they were real sequences.  All the tests were carried out on a machine with Debian 4.9, 252 GB of RAM and processor Intel(R) Xeon(R) Silver @ 2.10GHz, with 32 cores.

\section{Results}

The average compression rate achieved by our index is 1.89, meaning that, in all the cases, the data structure used about half the space of the plain representation of the reads (see Table~\ref{tab:dbg}). We also note that the smaller the value for $k$, the greater the size of the index. This behavior is expected as the dBG becomes denser when we decrease $k$. Thus, we have to store a higher number of colors per node.

\begin{table}[!htb]
\caption{Statistics about the different colored dBGs generated in the experiments. The index size is expressed in MB and considers the space of $BOSS(G)$ plus the space of our succinct version of $C$. The compression rate was calculated as the space of the plain representation of the reads (833.67 MB) divided by the index size.}
\label{tab:dbg}
\centering
\begin{tabular}{c|c|c|c|c|c}
\toprule
\textbf{k}&\textbf{Total number}&\textbf{Number of}  &\textbf{Number of}&\textbf{Index}    &\textbf{Compression} \\
          &\textbf{of nodes}    &\textbf{solid nodes}&\textbf{edges}    &\textbf{size}&\textbf{rate} \\
\midrule
25 & 106,028,714 & 11,257,781 & 120,610,151 & 446.38 & 1.86 \\
30 & 142,591,410 & 11,425,646 & 157,186,548 & 443.82 & 1.87 \\
35 & 179,167,289 & 11,561,630 & 193,773,251 & 441.18 & 1.88 \\
40 & 215,751,326 & 11,667,364 & 230,365,635 & 438.23 & 1.90 \\
45 & 252,337,929 & 11,743,320 & 266,958,709 & 435.30 & 1.91 \\
50 & 288,925,674 & 11,791,640 & 303,552,318 & 432.13 & 1.92 \\                       
\bottomrule
\end{tabular}
\end{table}

The number of colors of every instance is several orders of magnitude smaller than the number of reads, being $k=25$ the instance with more colors (6552) and $k=50$ the instance with the fewest (1689). Even though the fraction of colored nodes in every instance is small, the percentage of the index space used by the color matrix is still high (~73\% on average). Regarding the time for coloring the graph, it seems to be reasonable for practical purposes if we use several threads. In fact, building, filling and compacting $C$ took 5,015 seconds on average, and the value decreases if we increment $k$. The working space, however, is still considerable. We had memory peaks ranging from 3.03 GB to 4.3 GB, depending on the value for $k$ (see Table~\ref{tab:color}).

The process of reconstructing the reads yielded a small number of ambiguous sequences in all the instances (2,760 sequences on average), and decreases with higher values of $k$, especially for values above 40 (see Table~\ref{tab:color}).

\begin{table}[ht]
\centering
\caption{Statistics about our greedy coloring algorithm.  The column “Color space usage” refers to the percentage of the index space used by our succinct version of $C$. Elapsed time and memory peak are expressed in seconds and MB, respectively, and both consider only the process of building, filling, and compacting the color matrix.}
\label{tab:color}
\begin{tabular}{c|c|c|c|c|c|c}
\toprule
\textbf{k}&\textbf{Number of}    &\textbf{Number of}&\textbf{Color space}&\textbf{Ambiguous}&\textbf{Elapsed}&\textbf{Memory}\\
          &\textbf{colored nodes}&\textbf{colors}   &\textbf{usage}      &\textbf{sequences}&\textbf{Time}&\textbf{peak}\\
\midrule
25 & 21,882,874 & 6,552 & 83.03 & 1904 &5,835 & 4,391\\
30 & 21,907,324 & 4,944 & 79.14 & 1502 &5,551 & 4,119\\
35 & 21,926,687 & 2,924 & 75.27 & 1224 &5,131 & 3,847\\
40 & 21,942,083 & 2,064 & 71.40 & 1054 &4,872 & 3,575\\
45 & 21,954,138 & 1,888 & 67.51 & 714 &4,507 & 3,303\\
50 & 21,964,947 & 1,689 & 63.58 & 176 &4,199 & 3,030\\                       
\bottomrule
\end{tabular}
\end{table}

\section{Conclusions and further work}

Experimental results shows our data structure is succinct, and that has a practical use. Still, we believe that a more careful algorithm for constructing the index is still necessary to reduce the memory peaks during the coloring. Further compaction of the color matrix can be achieved by using more elaborated compression techniques. However, this extra compression can increase the construction time of the colored dBG and produce a considerable slow down in the algorithms that work on top of it for extracting information from the reads. Comparison of our results with other similar data structures is difficult for the moment. Most of the indexes based on colored dBGs were not designed to handle huge sets of colors like ours and the greedy recoloring of \cite{alipanahi2018recoloring} does not scale well and needs extra information for reconstructing the reads. Still, it is a promising approach that, with further work, can be used in the future as a base for performing Bioinformatics analyses in compressed space.

%
%
\bibliographystyle{splncs04}
\bibliography{main}

\begin{thebibliography}{10}
\providecommand{\url}[1]{\texttt{#1}}
\providecommand{\urlprefix}{URL }
\providecommand{\doi}[1]{https://doi.org/#1}

\bibitem{alipanahi2018recoloring}
Alipanahi, B., Kuhnle, A., Boucher, C.: Recoloring the colored de bruijn graph.
  In: Proc. 25th International Symposium on String Processing and Information
  Retrieval (SPIRE). pp. 1--11 (2018)

\bibitem{almodaresi2017rainbowfish}
Almodaresi, F., Pandey, P., Patro, R.: Rainbowfish: A succinct colored de
  {B}ruijn graph representation. In: Proc. 17th International Workshop on
  Algorithms in Bioinformatics (WABI). p. article 18 (2017)

\bibitem{bankevich2012spades}
Bankevich, A., Nurk, S., Antipov, D., Gurevich, A.A., Dvorkin, M., Kulikov,
  A.S., Lesin, V.M., Nikolenko, S.I., Pham, S., Prjibelski, A.D., et~al.:
  {SPA}des: A new genome assembly algorithm and its applications to single-cell
  sequencing. Journal of Computational Biology  \textbf{19}(5),  455--477
  (2012)

\bibitem{BBGPS15}
Boucher, C., Bowe, A., Gagie, T., Puglisi, S.J., Sadakane, K.: Variable-order
  de {B}ruijn graphs. In: Proc. 25th Data Compression Conference (DCC). pp.
  383--392 (2015)

\bibitem{BOSS12}
Bowe, A., Onodera, T., Sadakane, K., Shibuya, T.: Succinct de {B}ruijn graphs.
  In: Proc. 12th International Workshop on Algorithms in Bioinformatics (WABI).
  pp. 225--235 (2012)

\bibitem{bray2016near}
Bray, N., Pimentel, H., Melsted, P., Pachter, L.: Near-optimal probabilistic
  rna-seq quantification. Nature Biotechnology  \textbf{34}(5),  525--527
  (2016)

\bibitem{de1946combinatorial}
{d}e Bruijn, N.G.: A combinatorial problem. Koninklijke Nederlandse Akademie v.
  Wetenschappen  \textbf{49}(49),  758--764 (1946)

\bibitem{BW94}
Burrows, M., Wheeler, D.: A block sorting lossless data compression algorithm.
  Tech. Rep.~124, Digital Equipment Corporation (1994)

\bibitem{Cla96}
Clark, D.: Compact {PAT} Trees. Ph.D. thesis, University of Waterloo, Canada
  (1996)

\bibitem{elias1974}
Elias, P.: Efficient storage and retrieval by content and address of static
  files. Journal of the ACM  \textbf{21}(2),  246--260 (1974)

\bibitem{fano1971number}
Fano, R.M.: On the number of bits required to implement an associative memory.
  Massachusetts Institute of Technology (1971)

\bibitem{gbmp2014sea}
Gog, S., Beller, T., Moffat, A., Petri, M.: From theory to practice: Plug and
  play with succinct data structures. In: Proc. 13th International Symposium on
  Experimental Algorithms (SEA). pp. 326--337 (2014)

\bibitem{holley2015bloom}
Holley, G., Wittler, R., Stoye, J.: Bloom filter trie -- a data structure for
  pan-genome storage. In: Proc. 15th International Workshop on Algorithms in
  Bioinformatics (WABI). pp. 217--230 (2015)

\bibitem{idury1995new}
Idury, R.M., Waterman, M.S.: A new algorithm for {DNA} sequence assembly.
  Journal of Computational Biology  \textbf{2}(2),  291--306 (1995)

\bibitem{iqbal2012novo}
Iqbal, Z., Caccamo, M., Turner, I., Flicek, P., McVean, G.: De novo assembly
  and genotyping of variants using colored de {B}ruijn graphs. Nature Genetics
  \textbf{44}(2),  226--232 (2012)

\bibitem{kececioglu1995comb}
Kececioglu, J.D., Myers, E.W.: Combinatorial algorithms for {DNA} sequence
  assembly. Algorithmica  \textbf{13}(1),  7--51 (1995)

\bibitem{lewis2015guide}
Lewis, R.: A Guide to Graph Colouring. Springer (2015)

\bibitem{MN05}
M{\"a}kinen, V., Navarro, G.: Succinct suffix arrays based on run-length
  encoding. Nordic Journal of Computing  \textbf{12}(1),  40--66 (2005)

\bibitem{medvedev2007computability}
Medvedev, P., Georgiou, K., Myers, G., Brudno, M.: Computability of models for
  sequence assembly. In: Proc. 7th International Workshop on Algorithms in
  Bioinformatics (WABI). pp. 289--301. Springer (2007)

\bibitem{medvedev2011paired}
Medvedev, P., Pham, S., Chaisson, M., Tesler, G., Pevzner, P.: Paired de bruijn
  graphs: a novel approach for incorporating mate pair information into genome
  assemblers. Journal of Computational Biology  \textbf{18}(11),  1625--1634
  (2011)

\bibitem{mustafa2017metannot}
Mustafa, H., Kahles, A., Karasikov, M., Raetsch, G.: Metannot: A succinct data
  structure for compression of colors in dynamic de bruijn graphs. bioRxiv p.
  article 236711 (2017)

\bibitem{mustafa2018dynamic}
Mustafa, H., Schilken, I., Karasikov, M., Eickhoff, C., R{\"a}tsch, G., Kahles,
  A.: Dynamic compression schemes for graph coloring. Bioinformatics
  \textbf{35}(3),  407--414 (2018)

\bibitem{navarro2016compact}
Navarro, G.: Compact Data Structures: A Practical Approach. Cambridge
  University Press (2016)

\bibitem{okanohara2007practical}
Okanohara, D., Sadakane, K.: Practical entropy-compressed rank/select
  dictionary. In: Proc. 9th Workshop on Algorithm Engineering and Experiments
  (ALENEX). pp. 60--70 (2007)

\bibitem{pandey2018mantis}
Pandey, P., Almodaresi, F., Bender, M.A., Ferdman, M., Johnson, R., Patro, R.:
  Mantis: A fast, small, and exact large-scale sequence-search index. Cell
  Systems  \textbf{7}(2),  201--207 (2018)

\bibitem{raman2007}
Raman, R., Raman, V., Satti, S.R.: Succinct indexable dictionaries with
  applications to encoding k-ary trees, prefix sums and multisets. ACM
  Transactions on Algorithms  \textbf{3}(4),  article 43 (2007)

\bibitem{reuter2015high}
Reuter, J., Spacek, D., Snyder, M.: High-throughput sequencing technologies.
  Molecular Cell  \textbf{58}(4),  586--597 (2015)

\bibitem{salmela2016accurate}
Salmela, L., Walve, R., Rivals, E., Ukkonen, E.: Accurate self-correction of
  errors in long reads using de bruijn graphs. Bioinformatics  \textbf{33}(6),
  799--806 (2016)

\end{thebibliography}
\newpage
\section{Appendix}
\appendix
\setcounter{table}{0}
\setcounter{figure}{0}
\renewcommand{\thetable}{A\arabic{table}}
\renewcommand\thefigure{A.\arabic{figure}}

\section{Pseudocodes}

\begin{algorithm}[!htb]
\caption{Function \texttt{greedyCol}}\label{alg:gc}
\begin{algorithmic}[1]
\Procedure{{\rm\tt greedyCol}}{$G$,$N$,$R_i$,$M$} \Comment{$G$ is a dBG, $N$ is a bitmap, $R_i$ is a string and $M$ is array of lists}
\State{$R_i \gets \texttt{\$}R_i\texttt{\$}$} \Comment{append dummy symbols at the ends of $R_i$}
\State{$v \gets \texttt{string2node}(R_i[1..k-1])$}
\State{$W_i \gets \emptyset$}
\State$I_i \gets I_i \cup \texttt{rank}_1(N,v)$

\ForEach{$r \in R_i[k-1..|R_i|]$} \Comment{traverse the dBG path of $R_i$}
\State$o\gets \texttt{outdegree}(G,v)$
\If{$o>1$}
\For{$j\gets 1$ \textbf{to} $o$}
\State$I_i \gets I_i \cup \texttt{rank}_1(N,\texttt{forward\_r}(G,v,j))$
\EndFor
\EndIf

\State$i\gets \texttt{indegree}(G,v)$

\If{$i>1$}
\For{$j\gets 1$ \textbf{to} $i$}
\State$v' \gets \texttt{incomming\_r}(G,v,j)$
\State$o' \gets \texttt{outdegree}(G,v'))$
\If{$o'>1$}
\For{$j\gets 1$ \textbf{to} $o'$}
\State$I_i \gets I_i \cup \texttt{rank}_1(N,\texttt{forward\_r}(G,v',j))$
\EndFor
\EndIf
\EndFor
\EndIf

\If{$N[v]$ is \texttt{true}}
\State$W_i \gets W_i \cup \texttt{rank}_1(N,v)$
\EndIf
\State$v\gets \texttt{forward}(G,v,r)$
\EndFor

\State$W_i\gets W_i \cup \texttt{rank}_1(N,v)$
\State$I_i \gets I_i \cup \texttt{rank}_1(N,v)$

\ForEach{$n \in I_i$} \Comment{compute the colors already used}
\ForEach{$c \in M[n]$}
\State{$H_i[c]\gets$ \textbf{true}}
\EndFor
\EndFor
\State{$c' \gets$ minimum color not in $H_i$}
\ForEach{$n \in W_i$} \Comment{color the nodes}
\State{$M[n] \gets M[n] \cup c'$}
\EndFor
\EndProcedure
\end{algorithmic}
\end{algorithm}

\begin{algorithm}[!htb]
\caption{Function \texttt{buildSeqs}}\label{alg:bs}
\begin{algorithmic}[1]
\Procedure{{\rm\tt buildSeqs}}{$G$,$v$} \Comment{$G$ is a colored dBG and $v$ is a starting node}
\State$L \gets \emptyset$ \Comment{list of rebuilt sequences}
\State$A \gets$ \texttt{getColors($G$,$v$)}
\State$S \gets \texttt{nodeLabel}(G,v)$ \Comment{initialize an string with the label of $v$}
\ForEach{$a \in C$}\label{seqbuilt_for}
\State$v' \gets v$\Comment{temporal dBG node}
\State$S' \gets S, amb \gets$ \textbf{false} 
\While{\texttt{isEnding($G$,$v'$)} is \textbf{false} and \emph{amb} is \textbf{false}}
\State$o \gets$ \texttt{outdegree($G$,$v'$)}
\If{$o$ is $1$}
\State$S' \gets S' \cup \texttt{edgeSymbol}(G, v',1)$ \Comment{push the new symbol into $S'$}
\State$v' \gets \texttt{forward\_r}(G, 1)$
\Else
\State$m \gets 0$
\For{$u \gets 1$ \textbf{to} $o$} \Comment{check which successors of $v'$ has color $a$}
\If{$a \in \texttt{getColors}(\texttt{forward\_r}(G,v',u))$}
\State$v' \gets \texttt{forward\_r}(G,v',u)$
\State$m \gets m + 1$
\EndIf
\EndFor
\If{$m>1$} \Comment{more than one successor $v'$ has color $a$}
\State$amb \gets$ \texttt{true}  
\EndIf
\EndIf
\EndWhile

\If{$amb$ \textbf{not} \texttt{true}}
\State$L \gets L \cup S[2..|S|-1]$
\EndIf

\EndFor
\State\textbf{return} $L$
\EndProcedure
\end{algorithmic}
\end{algorithm}

\begin{algorithm}[ht]
\caption{Function \texttt{contigAssm}}\label{alg:ca}
\begin{algorithmic}[1]
\Procedure{{\rm\tt contigAssm}}{$G$,$v$,$x$} \Comment{$v$ is a starting node and $x$ is a threshold}

\State $L \gets \emptyset$ 
\State $S \gets \texttt{nodeLabel}(G,v)$

\ForEach{$c_i \in \texttt{getColors}(v)$}
\State $Q[c_i] \gets v$ 
\EndFor

\While{\textbf{true}}

\If{$\texttt{indegree}(G, v)>1$}
\State{$v' \gets \texttt{backward\_r}(G,v,1))$}
\If{$\texttt{isStarting}(v')$} \Comment{add new reads to the contig}
\ForEach{$c_i \in \texttt{getColors}(v')$}
\If{$L[(c_i,v')]$ is not \texttt{true}}
\State{$Q[c_i] \gets v'$}
\EndIf
\EndFor
\EndIf
\EndIf

\If{$o \gets \texttt{outdegree}(G,v) > 1$}
\State{$t \gets v, v \gets 0$}
\For{$i\gets1$ to $o$}\Comment{compute the most probable successor node}
\State{$v' \gets \texttt{forward\_r}(G,t,i)$}
\If{$\texttt{isEnding}(v')$}\Comment{discard reads ending at $v$}
\ForEach{$c_i \in \texttt{getColors}(v')$}
\State$L[(c_i,Q[c_i])] \gets \texttt{true}$
\EndFor
\State{$Q \gets Q \setminus A$}

\Else
\State{$A \gets \texttt{getColors}(v')$}
\State{$w \gets (Q \cap A)/|Q|$}\Comment{weight the successor node}
\If{$w \geq x$}
\State{$v\gets v'$}
\State{$Q \gets A$}
\State$S \gets S \cup \texttt{edgeSymbol}(G, t, i)$
\State{\textbf{break}}
\EndIf
\EndIf
\EndFor
\If{$v$ \textbf{is} 0} \textbf{break} \Comment{no successor has the minimum weight $x$}
\EndIf
\Else
\State $v\gets \texttt{forward\_r}(G, v,1)$
\If{$\texttt{isEnding}(v)$} \textbf{break}
\EndIf
\State$S \gets S \cup \texttt{edgeSymbol}(G, v, 1)$ 
\EndIf
\EndWhile
\State{\textbf{return} $S[2..|S|]$}
\EndProcedure
\end{algorithmic}
\end{algorithm}

\end{document}